\newif\ifconf
\newif\ifconf
\newif\ifconf
\conftrue
\ifconf
 \documentclass[letterpaper,conference,10pt]{ieeeconf}
 
 \newenvironment{IEEEkeywords}{\begin{keywords}}{\end{keywords}}
 \newcommand\IEEEQED\QED
 
\else
 \documentclass[letterpaper,10pt,journal,twoside]{IEEEtran}
\fi

\usepackage{float}
\usepackage{balance}
\usepackage{amsmath,mathtools,amssymb}
\usepackage{tikz,pgfplots}
\usepackage{tikz-3dplot}
\usetikzlibrary{calc,arrows,shapes,backgrounds,calc,positioning,patterns,decorations.pathmorphing,decorations.markings,mindmap,trees,fadings,angles,quotes}
\usepackage{blox}
\usepackage[european voltages,cute inductors,smartlabels]{circuitikz}
\usepackage{color}
\usepackage{siunitx}	
\usepackage{booktabs}
\usepackage{balance} 
\usepackage{xcolor}
\usepackage{cite}
\usepackage{bbm}
\usepackage{dsfont}

\usepackage{subcaption}
\usepackage{cases}
 \usepackage{caption}
\usepackage[T1]{fontenc}
 \usepackage{graphicx} 
\newtheorem{theorem}{Theorem}
\newtheorem{lemma}[theorem]{Lemma}

\newtheorem{remark}{Remark}[section]

\newtheorem{ass}{Assumption}
\newtheorem{definition}{Definition}

\usepackage{tabu}
\usepackage{bm}
\newcounter{cAss}
\newcounter{cAssSaved}
\newcommand\Ass[1]{\ensuremath{\boldsymbol{\mathcal A}_{\text{\hspace{0.75pt}\bf#1}}}}
\newlength\asswidth
 {\begin{list}{\Ass{\arabic{cAss}}:}{%
   \setcounter{cAssSaved}{\thecAss}\usecounter{cAss}\setcounter{cAss}{\thecAssSaved}
   \setlength\topsep{.667ex plus 2pt minus 2pt}
   \setlength\itemsep\topsep
   \settowidth\asswidth{\Ass{\arabic{cAss}}:}
   \addtolength\asswidth\labelsep
   \setlength\leftmargin\asswidth
   \setlength\labelwidth\asswidth}}%
 {\end{list}}
 
 \newcounter{cDef}
\newcounter{cDefSaved}
\newcommand\Def[1]{\ensuremath{\boldsymbol{\mathcal P}_{\text{\hspace{0.75pt}\bf#1}}}}
\newlength\defwidth
 {\begin{list}{\Def{\arabic{cDef}}:}{%
   \setcounter{cDefSaved}{\thecDef}\usecounter{cDef}\setcounter{cDef}{\thecDefSaved}
   \setlength\topsep{.667ex plus 2pt minus 2pt}
   \setlength\itemsep\topsep
   \settowidth\defwidth{\Def{\arabic{cDef}}:}
   \addtolength\defwidth\labelsep
   \setlength\leftmargin\defwidth
   \setlength\labelwidth\defwidth}}%
 {\end{list}}

\makeatletter
\xdef\@endgadget#1{{\unskip\nobreak\hfil\penalty50\hskip1em\hbox{}\nobreak\hfil#1\parfillskip=0pt\finalhyphendemerits=0\par}}
\newcommand\@Endofsymbol{$\triangledown$}
\newcommand\Endofremark{\@endgadget{\@Endofsymbol}}
\makeatother

\newcommand{\R}{\mathbb{R}}

\newcommand\E{\mathrm e}

\DeclareMathOperator\sign{sign}

\title{\LARGE \bf
Learning-Based Robust Fixed-Time Terminal Sliding Mode Control}

\author{Chaimae El Mortajine$^a$, Moussa Labbadi$^b$, Adnane Saoud$^c$ and Mostafa Bouzi$^a$ 
\thanks{$^a$C. El Mortajine and M. Bouzi are with the MIET Laboratory, Faculty of Science and Technology, Hassan First University of Settat, Settat 26000, Morocco (e-mail:
\textsl{c.elmortajine.doc@uhp.ac.ma,mostafa.bouzi@uhp.ac.ma}).}
\thanks{$^b$M. Labbadi is  with the Aix-Marseille University, LIS UMR CNRS 7020, Marseille, 
France (e-mail: \textsl{moussa.labbadi@lis-lab.fr}).}
\thanks{$^c$A. Saoud is with the College of Computing, University Mohammed VI
Polytechnic, Benguerir, Morocco (e-mail: \textsl{adnane.saoud@um6p.ma}).}}

\IEEEoverridecommandlockouts
\begin{document}

\maketitle
\thispagestyle{empty}

%%%%%%%%%%%%%%%%%%%%%%%%%%%%%%%%%%%%%%%%%%%%%%%%%%%%%%%%%%%%%%%%%%%%%%%%%%%%%%%%
\begin{abstract}
 In this paper, we develop and analyze an integral fixed-time sliding mode control method for a scenario in which the system model is only partially known, utilizing Gaussian processes. We present two theorems on fixed-time convergence. The first theorem addresses the fully known system model, while the second considers situations where the system’s drift is approximated utilizing Gaussian processes (GP) for approximating unknown dynamics. Both theorems establish the global fixed-time stability of the closed-loop system. The stability analysis is based on a straightforward quadratic Lyapunov function. Our proposed method outperforms an established adaptive fixed-time sliding mode control approach, especially when ample training data is available.
\end{abstract}


\section{Introduction}\label{Sec:I}

Sliding Mode Control (SMC) is renowned for its robustness in handling matched perturbations and ensuring finite-time stability \cite{utkin1977variable}. The SMC design involves defining a sliding surface and selecting a control law that drives the system trajectory to this surface, stabilizing it via a discontinuous control input. To resolve the reaching mode problem, Integral Sliding Mode Control (ISMC) was proposed, which eliminates the reaching phase and reduces control effort \cite{utkin1996integral}.

ISMC has gained significant attention in recent years, with research focusing on its theoretical and practical applications. While these studies typically guarantee asymptotic or exponential stability, which meets control objectives as time tends to infinity, this type of stability is often insufficient for time-critical systems \cite{bhat2000finite, moulay2006finite, efimov2021finite}. In systems like robotic networks and autonomous vehicles, finite-time stability is vital for synchronized, timely responses. Conventional stability concepts fail to address the dynamic interactions and time-sensitive responses of interconnected subsystems, highlighting the necessity for control strategies that guarantee finite-time stability. A critical challenge is that the settling time estimation is highly sensitive to initial conditions, which can be unpredictable. Polyakov \cite{polyakov2011nonlinear} addressed this by introducing fixed-time (FxT) stability, which ensures a settling time independent of initial conditions.

Recently, various approaches for FxT control have emerged. In  \cite{lopez2020finite} explored FxT control for double integrator systems, while in \cite{golestani2021fixed} examined high-order integral systems. FxT stabilization for autonomous systems was also studied by L'opez et al. \cite{lopez2019conditions}. Applications of these techniques appear in the works of Chen et al. \cite{chen2022adaptive}, Liu et al. \cite{liu2022overview}, Ni et al. \cite{ni2016fast}, and Zou et al. \cite{zou2020fixed, zou2019fixed}. The integration of ISMC with FxT stability has recently been developed \cite{khanzadeh2023fixed}, offering the advantage of eliminating the reaching mode and ensuring a fixed-time property. However, these approaches often involve complex control designs with numerous parameters, making tuning challenging. Furthermore, robustness against perturbations remains an issue, signaling the need for further improvements in robustness and simplicity for practical applications.

Gaussian processes (GPs) have become widely used in control theory \cite{hashimoto2022learning, capone2019backstepping} due to their flexible, nonparametric, data-driven framework that balances data fitting and regularization in noisy environments \cite{umlauft2017feedback, umlauft2018uncertainty, chowdhary2014bayesian, seeger2004gaussian, doyle2002identification, ljung1999system}. GPs also quantify model uncertainty, enabling the derivation of model error bounds, unlike traditional methods like Volterra series or artificial neural networks, which require careful model structure selection and do not inherently account for uncertainty \cite{ljung1999system}.

In this paper, we propose an integral fixed-time sliding mode control method for nonlinear systems \cite{labbadi2023design}. The method introduces a non-singular terminal sliding variable with a guaranteed fixed-time property. A feedback control law is then applied to achieve fixed-time stabilization of the system under matched perturbations, with easily tunable parameters. For systems with partially unknown dynamics, GPs are employed to estimate the drift term and control effectiveness, aiding in the design of the equivalent control. Stability and settling time estimates for the closed-loop system are provided, alongside numerical simulations to demonstrate the effectiveness of the approach.

\section{Preliminaries}\label{Sec:II}
\subsection{Fixed- and finite-time stability}
{\textcolor{black}{To review some findings regarding fixed- and finite-time stability, Consider the following system:}
\begin{equation}
\label{eq:fz}
\begin{aligned}
\dot{x}(t) = f(x(t)), \quad x(t) \in \mathbb{R}^n, \quad x(0) = x_0.
\end{aligned}
\end{equation}
with the continuous function $f$.
\begin{definition}
From \cite{bhat2000finite}: {\textcolor{black}{a}} system denoted as \eqref{eq:fz} is considered globally finite-time stable if it satisfies Lyapunov stability and for any initial condition vector $x_0$ in $\mathbb{R}^n$, there exists a settling time function $T(x_0) \geq 0$ dependent on the initial conditions. Specifically, for any solution $x(\cdot)$ of system \eqref{eq:fz} with the initial condition $x(0) = x_0$, it holds that as $t$ approaches $T(x_0)$, the norm of $x(t)$ converges to zero, i.e., $\lim_{t \to T(x_0)} |x(t)| = 0$. In other words, the norm of $x(t)$ remains identically zero for all $t \geq T(x_0)$. The function $T(x_0)$ is referred to as the settling time.
\end{definition}
\begin{definition} {\textcolor{black}{(see \cite{polyakov2011nonlinear}): System \eqref{eq:fz} is considered globally FxT stable if it satisfies the following conditions: 
1) It is globally finite-time stable; 2) The settling-time function \(T(x_0)\) is upper-bounded by a constant \(T_{\max} > 0\), meaning that for all \(x_0 \in \mathbb{R}^n\), \(T(x_0) \leq T_{\max}\), and \(T_{\max}\) remains constant regardless of the initial conditions.
}}
\end{definition}
The following lemmas are necessary for the subsequent analysis.
\begin{lemma}\label{lemma:1} (see \cite{polyakov2011nonlinear}): If there exists a continuously differentiable positive definite radially unbounded function $V: \mathbb{R}^n \to \mathbb{R}^+$ such that
\begin{small}
\[
\dot{V}(x) \leq -c_1V(x)^{h_{\textcolor{black}1}} - c_2V(x)^{h_2} 
\]\end{small}
where $V$ is called the Lyapuov function and $x \in \mathbb{R}^n$, $c_i> 0$, and $0 < h_1 < 1 < h_2$, then system \eqref{eq:fz} is globally FxT stable  {\textcolor{black}{with a settling time bounded by}}
\begin{small}
\begin{align}
T(x_0) \leq \frac{1} {c_1(1 - h_1)} + \frac{1} {c_2(h_2 - 1)}.
\end{align}\end{small}
\end{lemma}
\begin{lemma}\label{lemma:2}\cite{labbadi2023design}
Consider the following system:
\begin{equation}
\dot{x}(t)=-\alpha\frac{\sqrt{\pi}}{2}\exp(x^2)\sign(x)+d(t),\quad x(0)=x_0,\label{eq:sysP}
\end{equation}
where \(\alpha>\tfrac{2}{\sqrt{\pi}}\bar{d}\),   $x\in\R$, and  $d\in\R$.
Then,  system \eqref{eq:sysP} is globally fixed-time stable with a settling time $T(x_0)\leq T_{\max}$ where
\(T_{\max}=\frac{1}{\alpha-\tfrac{2}{\sqrt{\pi}}\bar{d}}.
\)
\end{lemma}
\subsection{Problem statement}
Consider the following perturbed system \cite{khanzadeh2023fixed}.
\begin{equation}
\begin{cases}
\dot{x}_{\textcolor{black}{1}}(t) = f_1(x) + g_1(x)u_1(t) + d_1(t) \\
\dot{x}_{\textcolor{black}{2}} (t) = f_2(x) + g_2(x)u_2(t) + d_2(t) \\
\vdots \\
\dot{x}_n(t) = f_n(x) + g_n(x)u_n(t) + d_n(t)
\end{cases}
\label{eq:model}
\end{equation}
where: 
{\textcolor{black}{$x = [x_1 \ x_2 \ \cdots \ x_n]^T \in \mathbb{R}^n$ is the state,
$u = [u_1 \ u_2 \ \cdots \ u_n]^T \in \mathbb{R}^n$ is the input,
$f_i$ and $g_i$ are known nonlinear functions such that $g_i\neq 0$ for all $x \in \mathbb{R}^n$.
$d_i \in \mathbb{R}^n$ is a perturbation assumed to be bounded but with an unknown upper bound, satisfying $|d_i(t)| \leq \bar{d}_i$ with $(i = \{1, \ldots, n\})$.}} The mathematical model represented by equation \eqref{eq:model} holds significance across various crucial engineering applications, including but not limited to multi-agent systems \cite{defoort2015leader}, synchronization of complex dynamical networks \cite{yang2017fixed}, recurrent neural networks \cite{yang2017fixed}, and beyond. 
This form is commonly found in many practical systems such as robot manipulators, autonomous vehicles, and quadrotors. These systems often exhibit complex dynamics and uncertainties that must be managed effectively to ensure safe and reliable operation. To ensure the applicability of our approach, we make the following assumption regarding the matrix \( g \).

\section{Fixed-time tracking with know dynamics}\label{sec:FxT}
In this section, we present the primary findings of the letter. Inspired by the methodology detailed in \cite{labbadi2023design}, we introduce a nonsingular FxT sliding variable. Subsequently, we utilize a robust FxT reaching controller to effectively stabilize the sliding variable within an FxT framework.
Firstly, let's define the tracking error as follows:
\begin{equation}\label{eq:e}
	\color{black}{{z_i(t)} = {x}_i(t) - {x_{id}(t)}}
	\end{equation}
where $x_i(t)$ is the output of the system \eqref{eq:model}, and $\color{black}{x_{id}}$ denotes its reference value. Inspired by \cite{moulay2021robust}, we introduce the following sliding variable:
	\begin{equation}\label{eq:s}
	\begin{aligned}
	s_i(z_i) = {z_i (t)} + \alpha_{i1}\int_{0}^t \exp(z_{i}(\tau)^2)\left\lfloor z_i(\tau)\right\rceil^{\frac{p_i}{q_i}} d\tau,
	\end{aligned}
	\end{equation}
with \(\alpha_{i1}>0\)   and \(\frac{p_i}{q_i}\in[0 \ 1)\). We propose the following control law:
\begin{small}
\begin{equation}\label{eq:u}
    \begin{aligned}
        u_i(z) &= -g_i(x)^{-1} \left(f_i(x)+ \alpha_{i1}\exp(z_{i}(t)^2)\left\lfloor z_i(t)\right\rceil^{\frac{p_i}{q_i}} \right.\\ 
    & \quad \left.- {\dot{x}_{id}(t)}+\frac{\sqrt{\pi}}{2}\alpha_{i2}\exp(s_i(z_i)^2)\left\lfloor s_i(z_i)\right\rceil^{0}\right)
    \end{aligned}
\end{equation}
\end{small}
 with $\alpha_{i2} >  \frac{2}{\sqrt{\pi}}\bar{d}_i>0$.
We present the following theorem.
\begin{theorem}\label{theorem:1}
   If \( \alpha_{i2} >  \bar{d}_i \), the global closed-loop system \eqref{eq:model}--\eqref{eq:u} is FxT stable with a  bounded settling time satisfying  
   \begin{equation}\label{eq:T}
     T_{\max}= T(s) + T(z),
   \end{equation}
where 
\begin{equation}\label{eq:T1}T(z)\leq \max\bigg(T_1(z_1), T_2(z_2),\dots, T_n(z_n)\bigg).\end{equation}
with 
\begin{equation}\label{eq:T2}
T_i(z_i) \leq \frac{2}{\alpha_{i1}}.\frac{1}{1-\frac{p_i^2}{q_i^2}}\end{equation}
and 
\begin{small}
\begin{equation}\label{eq:T3} T(s) \leq \max\bigg(T_1(s_1), T_2(s_2),\dots, T_n(s_n)\bigg),\end{equation}\end{small}
with
\begin{small}
\begin{equation}\label{eq:T4}
T_i(s_i) \leq \frac{1}{\alpha_{i2}-\tfrac{2}{\sqrt{\pi}}\bar{d}_i}.
\end{equation}
\end{small}

\end{theorem}
\begin{proof}
\textcolor{black}{Consider the sliding manifold 
\[
\mathcal{S}=\left\{z_i(t)\in\mathbb{R}:s_i(z_i) = 0\right\}.
\]}
   The derivative of $s_i$ is,
\begin{equation}
    \begin{aligned}
        \dot{s}_i(z_i) &=  \alpha_{i1}\exp(z_{i}(t)^2)\left\lfloor z_i(t)\right\rceil^{\frac{p_i}{q_i}}- {\dot{x}_{id}(t)} +g_i(x)u_i(z)\\ 
    & +\alpha_{i1}\exp(z_{i}(t)^2)\left\lfloor z_i(t)\right\rceil^{\frac{p_i}{q_i}}+f_i(x)+d_i(t),
    \end{aligned}
\end{equation}

    Substituting the input, yields
\[\dot{s}_i(z_i)= -\frac{\sqrt{\pi}}{2}\alpha_{i2}\exp(s_i(z_i)^2)\left\lfloor s_i(z_i)\right\rceil^{0}+d_i(t),\]
Applying Lemma \ref{lemma:2}, the sliding variable \eqref{eq:s} is FxT stable with a bounded settling time $T_i(s_i)$ as given in \eqref{eq:T4} and for  system \eqref{eq:model}, the global settling time is bounded by $T(s)$ as presented in \eqref{eq:T3}.

Let us now prove that when \( z_i \in \mathcal{S} \), from \eqref{eq:s}, we have
\[
\dot{z}_i(t) =  -\alpha_{i1}\exp(z_{i}(t)^2)\left\lfloor z_i(t)\right\rceil^{\frac{p_i}{q_i}},\]
Let us define the Lyapunov function as: \( V_i(z_i) = z_i^2(t) \), one has
\begin{equation}\label{eq:LF_z}
    \begin{aligned}
\dot{V}_i = 2  z_i\dot{z}_i&=-2z_i\alpha_{i1}\exp(z_{i}(t)^2)\left\lfloor z_i(t)\right\rceil^{\frac{p_i}{q_i}}\\
&=-2\alpha_{i1}\exp(z_{i}(t)^2)|z_i(t)|^{\frac{p_i}{q_i}+1},  
    \end{aligned}
\end{equation}
The derivative of Lyapunov function admits two dynamics with respect to $V_i$ and $|z_i|$.

\emph{Case 1}: When $|z_i|\geq 1$ and  $V_i(z_i)\geq 1.$ 
We have 
\[\exp(z_{i}(t)^2)|z_i(t)|^{\frac{p_i}{q_i}+1} > z_{i}(t)^2|z_i(t)|^{\frac{p_i}{q_i}+1} = |z_i(t)|^{\frac{p_i}{q_i}+3}\]
From \eqref{eq:LF_z}, one has 
\begin{equation}\label{eq:LF_z1}
    \begin{aligned}
\dot{V}_i &\leq -2\alpha_{i1}|z_i(t)|^{\frac{p_i}{q_i}+3}= -2\alpha_{i1}V_i(z_i)^{\frac{\frac{p_i}{q_i}+3}{2}}
\end{aligned}
\end{equation} 
Applying Lemma \ref{lemma:1}, with $h_2 = \frac{\frac{p_i}{q_i}+3}{2}$ and $c_2= 2\alpha_{i1}>0$, we conclude that the tracking error is FxT stable, with a settling time. \[T_i^{'} \leq \frac{1}{\alpha_{i1}(\frac{p_i}{q_i}+1)},\]

\emph{Case 2.} When $V_i(z_i)<1$ and $|z_i|<1$, one has 
\begin{equation}\label{eq:LF_z2}
    \begin{aligned}
\dot{V}_i  &= -2\alpha_{i1}\exp(z_{i}(t)^2)|z_i(t)|^{\frac{p_i}{q_i}+1}\\
&\leq -2\alpha_{i1}|z_i(t)|^{\frac{p_i}{q_i}+1}= -2\alpha_{i1}V_i(z_i)^{\frac{\frac{p_i}{q_i}+1}{2}}
\end{aligned}
\end{equation} 
Due to  \(\min\left(\exp(z_{i}(t)^2)\right)=1 .\)
Applying Lemma \ref{lemma:1} such that $h_1 = \frac{1+\frac{p_i}{q_i}}{2}<1$ due to $\frac{p_i}{q_i}<1$ and $c_1= 2\alpha_{i1}>0$, we conclude that the tracking error is FxT stable, with a settling time.
 
\begin{align}
T_i^{''} \leq \frac{1}{\alpha_{i1}(1-\frac{p_i}{q_i})},
\end{align}
Then, the tracking error is FxT stable with a bounded settling time $T_i(z_i) \leq T_i^{'}+ T_i^{''}$ as given in \eqref{eq:T2}. Additionally, all tracking errors can converge with a settling time $T(z)$ as presented in \eqref{eq:T1}.

Finally, the closed-loop system is globally FxT stable with a settling time $T_{\text{max}}= T(s) + T(z)$ given by \eqref{eq:T}. This completes the proof.
 
\end{proof}
\begin{remark}
As stated previously, the proposed method has the potential to achieve robustness with a single parameter, especially in integral sliding mode control with fixed-time convergence, compared to fixed-time controllers proposed in \cite{khanzadeh2023fixed}. Moreover, the work in \cite{khanzadeh2023fixed} does not include perturbations in the control design; the authors considered only free perturbations. However, our approach is robust against matched perturbations.
\end{remark}
\begin{remark}
For setting the control parameters, only three parameters are required, in contrast to the approach developed in \cite{khanzadeh2023fixed}. Our method simplifies tuning to achieve fixed-time convergence.
\end{remark}

The estimated
models $\hat{f}_i$ are then given by \cite{hashimoto2022learning}
\begin{align}\label{eq:f_estimated}
   \hat{f}_i = \begin{bmatrix} \varrho_{1,i}(x_i) & \cdots & \varrho_{n,i}(x_i) \end{bmatrix}^T, 
\end{align}
where $\varrho_{1,i}(x_i), \ldots, \varrho_{n,i}(x_i)$ are the means of the GPs.
\section{Learning-Based Fixed-Time Control with Gaussian Process Regression}

To address the problem of uncertainty approximation, we employ Gaussian Process (GP) regression as a data-driven method to model the unknown function \( f_i(x) \) and its associated uncertainty. Using GP regression, we propose a learning-based fixed-time controller that guarantees fixed-time convergence.

An important assumption imposes a restriction on the complexity of the map \( f_i(x) \) through the reproducing kernel Hilbert space (RKHS) norm, as described below.

\begin{ass}\label{ass:3}
    For the map \( f_i: X \rightarrow \mathbb{R}^n \) in \( \mathcal{S} \), the RKHS norm with respect to the kernel \( \kappa \) is bounded, i.e., \( \|f_i\|_\kappa \leq \infty \) for all \( i \in \{1, \ldots, n\} \).
\end{ass}

This assumption is satisfied by most commonly used kernels, as all continuous functions defined over a compact state-space meet this condition \cite{seeger2004gaussian}. For further details on the RKHS norm, we refer the interested reader to \cite{case2019note}.

We assume access to measurements \( x \in X \) and \( y_i = f_i(x) + \mathcal{W} \), where \( \mathcal{W} \sim \mathcal{N}(0_n, \rho_F^2 \bm{I}_n) \) represents additive noise with \( \rho_F \in \mathbb{R}^+_0 \), as stated in \cite{umlauft2018uncertainty}.

\begin{ass}\label{ass:4}
    Measurements \( x \in X \) and \( y_i = f_i(x) + \mathcal{W} \) are available, where \( \mathcal{W} \sim \mathcal{N}(0_n, \rho_F^2 \bm{I}_n) \) is the noise.
\end{ass}

\subsection{Gaussian Process Regression}
Gaussian Process (GP) regression is a flexible tool for learning unknown functions from noisy observations. We model the function \( f(x) \) as a vector of components: \( F(x) = [F_1(x), \dots, F_m(x)]^\top \), where each \( F_i(x) \) is a zero-mean GP with a covariance function \( \kappa_i(x, x') \). The covariance function encapsulates the smoothness and relationships between data points. GP regression computes the posterior distribution of \( f(x) \) at any point \( x \), providing its mean as the most likely estimate and its variance as a measure of uncertainty. Given \( N \) measurements \( \{x^{(1)}, \dots, x^{(N)}\} \) and the corresponding noisy outputs \( \{y^{(1)}, \dots, y^{(N)}\} \), where \( y^{(j)} = f(x^{(j)}) + w^{(j)} \), the posterior distribution of \( F_i(x) \) is Gaussian, with the following mean and covariance:
\[
\varrho_i(x) = \overline{\kappa}_i^T (\kappa_i + \sigma_F^2 \bm{I}_N)^{-1} y_i
\]
\[
\sigma_i^2(x) = \kappa_i(x, x) - \overline{\kappa}_i^T (\kappa_i + \sigma_F^2 \bm{I}_N)^{-1} \overline{\kappa}_i
\]
Here, \( \overline{\kappa}_i \) is the vector of covariances between the test point \( x \) and the training points, while \( \kappa_i \) is the covariance matrix of the training points. The variance \( \sigma_i^2(x) \) represents the uncertainty in the prediction at the point \( x \).
The approximation error at any point \( x \in X \) is bounded as follows:
\[
|f_i(x) - \varrho_i(x)| \leq \tilde{\chi}_i \sigma_i(x)
\]
where \( \tilde{\chi}_i \) is a constant. This bound ensures that the approximation error remains small with high probability. For computational efficiency, we use multiple scalar GPs to model each component of \( f(x) \), avoiding the high computational cost of multivariate GP formulations \cite{capone2019backstepping}. The function \( f(x) \) is approximated by the vector of means:
\[
\varrho(x) = [\varrho_1(x), \dots, \varrho_m(x)]^T
\]
and the vector of variances:
\[
\sigma^2(x) = [\sigma_1^2(x), \dots, \sigma_m^2(x)]^T
\]

This approach provides a computationally efficient and robust solution for learning-based control, enabling effective handling of model uncertainties.

\subsection{Fixed-time control based on GPs}
A FxT-SM control design approach utilizing GP models is introduced. The control law is formulated directly using the development given in the previous section. 

 Initially, the model estimates $\hat{f}_i$ are calculated using pre-collected data before the control design phase. The relevant formulations are provided in \eqref{eq:f_estimated}. 
 Then, the control law \eqref{eq:u} is modified by:
 We propose the following control law:
\begin{small}
\begin{equation}\label{eq:u_GP}
    \begin{aligned}
        u_i(z) &= -g_i(x)^{-1} \left(\hat{f}_i(x)+ \alpha_{i1}\exp(z_{i}(t)^2)\left\lfloor z_i(t)\right\rceil^{\frac{p_i}{q_i}} \right.\\ 
    & \quad \left.- {\dot{x}_{id}(t)}+\alpha_{i2}\exp(s_i(z_i)^2)\left\lfloor s_i(z_i)\right\rceil^{0}\right)
    \end{aligned}
\end{equation}
\end{small}
 with $\alpha_{i2} >  \frac{2}{\sqrt{\pi}}\bar{d}_i>0$. 
 We need on the following Lemma to complete Theorem \ref{theorem:2}. 
\begin{lemma}
    \label{lemma:Lyapunov_ex}
    Consider a Lyapunov function \(V(z)=\frac{1}{2}z^2\) and its derivative satisfying \(\dot{V}(z) \leq
        \displaystyle |z|A\E^{z^2},\)
        where  $A<0$. Then, the variable \(z\) is fixed-time stable and settling time satisfying \(T(z)<T_{max}\) with,
        \(T_{max} =- \frac{2\sqrt{2}+1}{2A},\)
\end{lemma}
\begin{proof}
    The proof contains two cases.

    \emph{Case 1}: when \(|z|>1\), we have $\E^{z^2}>z^2$, therefore,
     \[\dot{V}(z) \leq
        Az^3
        =2.2^{\frac{1}{2}}AV^{\frac{3}{2}},\]
       Applying Lemma 1, the settling time is,
 \[T_1(z)\leq -\frac{1}{2A},\]
\emph{Case 2}: when \(|z|\leq 1\), we have $\min(\E^{z^2})\geq 1$, therefore,
 \[\dot{V}(z) \leq
        A|z|
        =2^{\frac{1}{2}}AV^{\frac{1}{2}},\]
        Applying Lemma 1, the settling time in this case  is,
        \[T_2(z)\leq -\frac{\sqrt{2}}{A},\]
        Then, the ultimate settling time is $T(z) = T_1(z)+ T_2(z) \leq T_{\max}$, which completes the proof. 
\end{proof}
 % =================================================================================================================================================================================================================
 % ++
 \begin{theorem}\label{theorem:2}
  Under both Assumptions 1 and 2 and  if \( \alpha_{i2} >  \bar{d}_i + |\Delta f_i(x)| \), the global closed-loop system \eqref{eq:model}--\eqref{eq:u_GP} is FxT stable with a  bounded settling time satisfying  
   \begin{equation}\label{eq:T_GP}
     T_{\max}= T(s) + T(z),
   \end{equation}
where 
\begin{equation}\label{eq:T5}T(z)\leq \max\bigg(T_1(z_1), T_2(z_2),\dots, T_n(z_n)\bigg).\end{equation}
with 
\begin{equation}\label{eq:T6}
T_i(z_i) \leq \frac{2}{\alpha_{i1}}.\frac{1}{1-\frac{p_i^2}{q_i^2}}\end{equation}
and 
\begin{small}
\begin{equation}\label{eq:T7} T(s) \leq \max\bigg(T_1(s_1), T_2(s_2),\dots, T_n(s_n)\bigg),\end{equation}\end{small}
with
\begin{small}
\begin{equation}\label{eq:T8}
T_i(s_i) \leq \frac{2\sqrt{2}}{2(\alpha_{i2}-\bar{d}_i - |\Delta f_i(x)|}).
\end{equation}
\end{small}

\end{theorem}
\begin{proof}
Consider the sliding manifold 
\[
\mathcal{S}=\left\{z_i(t)\in\mathbb{R}:s_i(z_i) = 0\right\}.
\]

  Substituting the input \eqref{eq:u_GP}, the dynamics of the sliding mode can be expressed as:
 \begin{small}
     \begin{equation}
        \dot{s}_i(z_i) =  \Delta f_i(x) -\alpha_{i2}\exp(s_i(z_i)^2)\left\lfloor s_i(z_i)\right\rceil^{0}+d_i(t).
\end{equation}
\end{small}
where
$\Delta f_i(x) = f_i(x) - \hat{f}_i(x)$.
    Consider the following Lyapunov candidate function:
\begin{align}\label{eq:LF_GP}
    V(s_i) = \frac{1}{2}s_i(z_i)^2,
\end{align}
Whose derivative is, 
\begin{align}\label{eq:D_LF_GP}
\begin{aligned}
    \dot{V}(s_i) &= s_i(z_i)\dot{s}_i(z_i)\\
    &= s_i(z_i) \left( - \alpha_{i2} \exp(s_i(z_i)^2) \left\lfloor s_i(z_i) \right\rceil^{0} \right.\\
    &\quad \left. + \Delta f_i(x) + d_i(t) \right)
\end{aligned}
\end{align}

Both perturbation \(d_i\) and \(\Delta f_i(x)\) are bounded according to Assumptions \ref{ass:3} and \ref{ass:4}. Then, we have the following dynamics.
\begin{align}\label{eq:D_LF_GP1}
\begin{aligned}
    \dot{V}(s_i) &\leq    - \alpha_{i2} \exp(s_i(z_i)^2)  |s_i(z_i)|\\
    &\quad   + |s_i(z_i)|(|\Delta f_i(x)| + \bar{d}_i) 
\end{aligned}
\end{align}
We have $\exp{(s_i(z_i)^2)} \geq 1$, indicating that the exponential of the square of $s_i(z_i)$ is at least 1. Therefore, the result stands as:

\begin{align}\label{eq:D_LF_GP2}
\begin{aligned}
    \dot{V}(s_i) &\leq    -\left(\alpha_{i2} - |\Delta f_i(x)| - \bar{d}_i\right) \exp(s_i(z_i)^2)  |s_i(z_i)|
\end{aligned}
\end{align}
Hence, $\dot{V}(s_i) < 0$ is ensured when $\alpha_{i2} > |\Delta f_i(x)| + \bar{d}_i$.
The magnitude of $|\Delta f_i(x)|$ is bounded. This is corroborated by Lemma 5, which guarantees that $|\Delta f_i(x)|$ holds with a probability of at least $1 - \vartheta$, thereby affirming the desired result.

Consequently, by applying \ref{lemma:Lyapunov_ex} with $A = -\left(\alpha_{i2} - |\Delta f_i(x)| - \bar{d}_i\right)$, the sliding variable \eqref{eq:s} is established to be FxT stable, exhibiting a bounded settling time $T_i(s_i)$ as delineated in \eqref{eq:T8}. Moreover, for system \eqref{eq:sysP}, the global settling time is constrained by $T(s)$ as articulated in \eqref{eq:T7}.

On the sliding manifold, denoted as $z_i \in \mathcal{S}$, the settling time remains consistent as per \eqref{eq:T2}. Consequently, the tracking error is demonstrated to be FxT stable, featuring a bounded settling time $T_i(z_i) \leq T_i^{'} + T_i^{''}$ as detailed in \eqref{eq:T6}. Additionally, all tracking errors can converge with a settling time $T(z)$ as explicated in \eqref{eq:T5}.

In conclusion, the closed-loop system is globally FxT stable with a settling time $T_{\text{max}} = T(s) + T(z)$ as stipulated in \eqref{eq:T_GP}. This concludes the proof.
\end{proof}
\begin{remark}
    It is important to note that in \eqref{eq:s}, the bound of the sliding mode variable can be minimized by choosing high control gains $\alpha_{i2} > 0$, although this may result in higher control inputs. Furthermore, with $\alpha_{i2} > 0$, the radius of \eqref{eq:s} in online learning scales with $\sigma(x)$. Consequently, a small $\sigma(x)$ ensures that the proposed online learning method consistently achieves a small sliding mode variable bound. In scenarios where $\Delta f_i$ approaches zero, the settling time matches that proposed in Theorem \ref{theorem:1}.

\end{remark}

\section{Simulations}
We consider a permanent magnet synchronous motor (PMSM)  \cite{khanzadeh2023fixed}:
\[
\begin{aligned}
\dot{x}_1 &= 2.5(x_2 - x_1) +u_1+d_1 (t) \\
\dot{x}_2 &= -x_2 - x_3x_1 + 25 x_1 +u_2+d_2 (t)\\
\dot{x}_3 &= -x_3 + x_1x_2+u_3+d_3 (t)
\end{aligned}
\]
withe the perturbations are: $
    d_1(t) = \sin(10 t), \quad 
    d_2( t) = \cos(10 t) 
    d_3(t) = \cos(10 t)  \sin(4 t).$
Consider the following control parameters.
$\alpha_{i1}  = 6, \alpha_{i2} = 4,  p_i = 8, q_i  = 10,  \bar{d}_i  = 1.$ with 
$T_i(x_i) = 0.92593, T_i(s_i) = 2.8284, T(x) = 0.92593, 
T(s) = 0.57364,   T_{\max} = 3.7544$.
The known and unknown components of the dynamics are the same as in \cite{khanzadeh2023fixed}. The GP hyperparameters were chosen by taking prior knowledge of the system into account.  
The control law is employed using two different model estimates, one trained using \( N = 5 \) data points per GP, and another using \( N = 50 \) data points per GP.  For the GPs, squared-exponential kernels are employed \( k(x, x^{'}) = \exp(-l |x - x^{'}|) \) with \( l = 1 \).
\begin{figure}[!ht]
\centering
\includegraphics[scale=0.44]{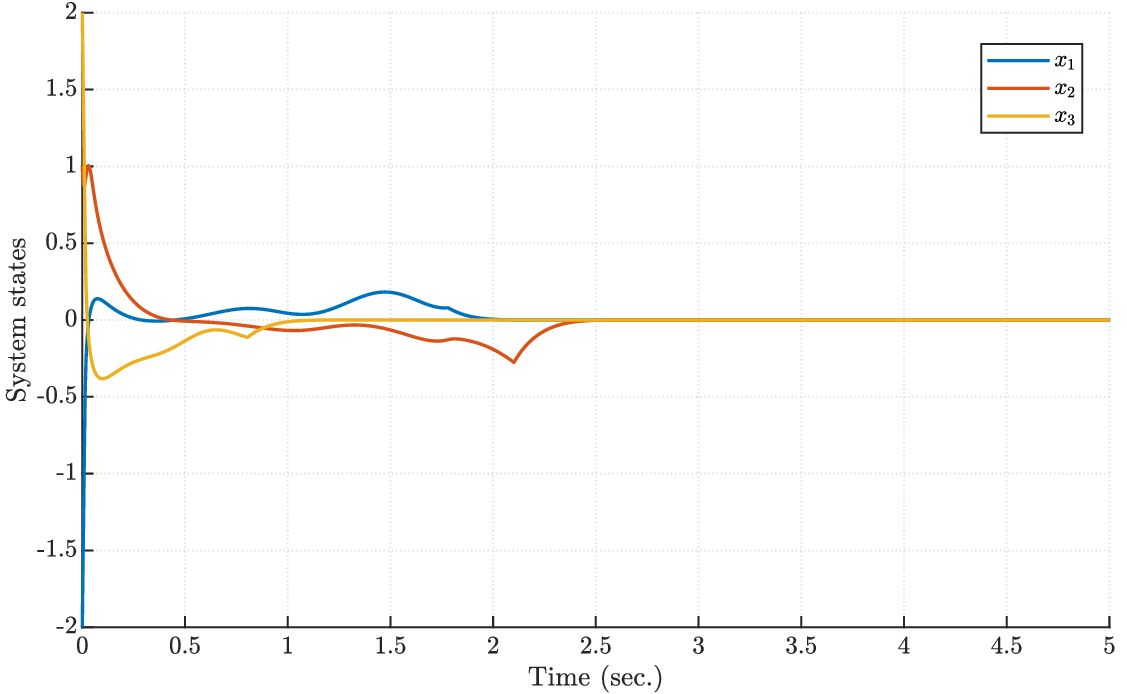}   
\caption{State variables   with $N= 5$ data.} 
\label{fig:States_GP_5}
\end{figure}
\begin{figure}[!ht]
\centering
\includegraphics[scale=0.44]{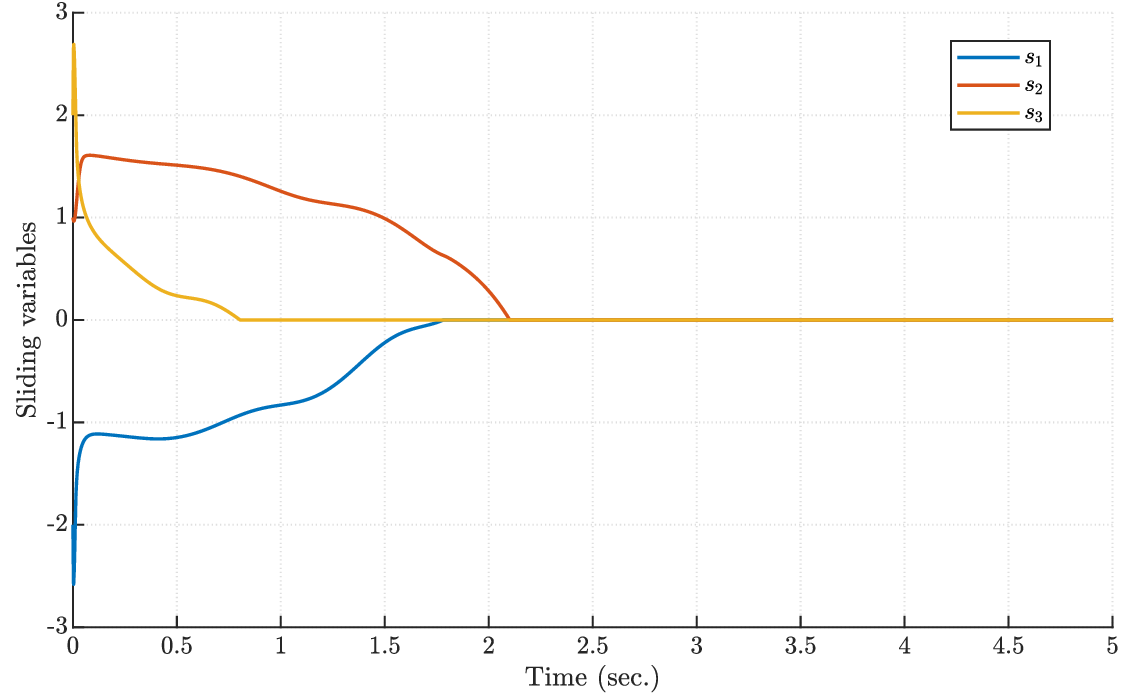}   
\caption{Sliding variables  with $N= 5$ data.} 
\label{fig:Sliding_GP_5}
\end{figure}
\begin{figure}[!ht]
\centering
\includegraphics[scale=0.44]{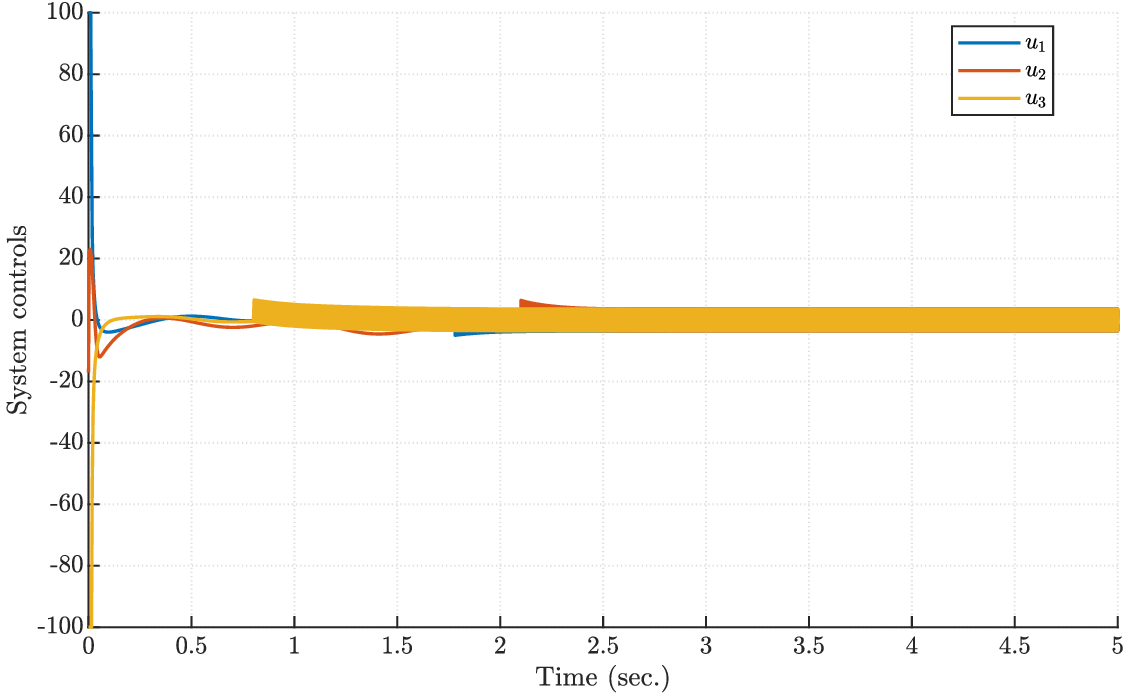}   
\caption{Input control  with $N= 5$ data.} 
\label{fig:Inputs_GP_5}
\end{figure}
\begin{figure}[!ht]
\centering
\includegraphics[scale=0.44]{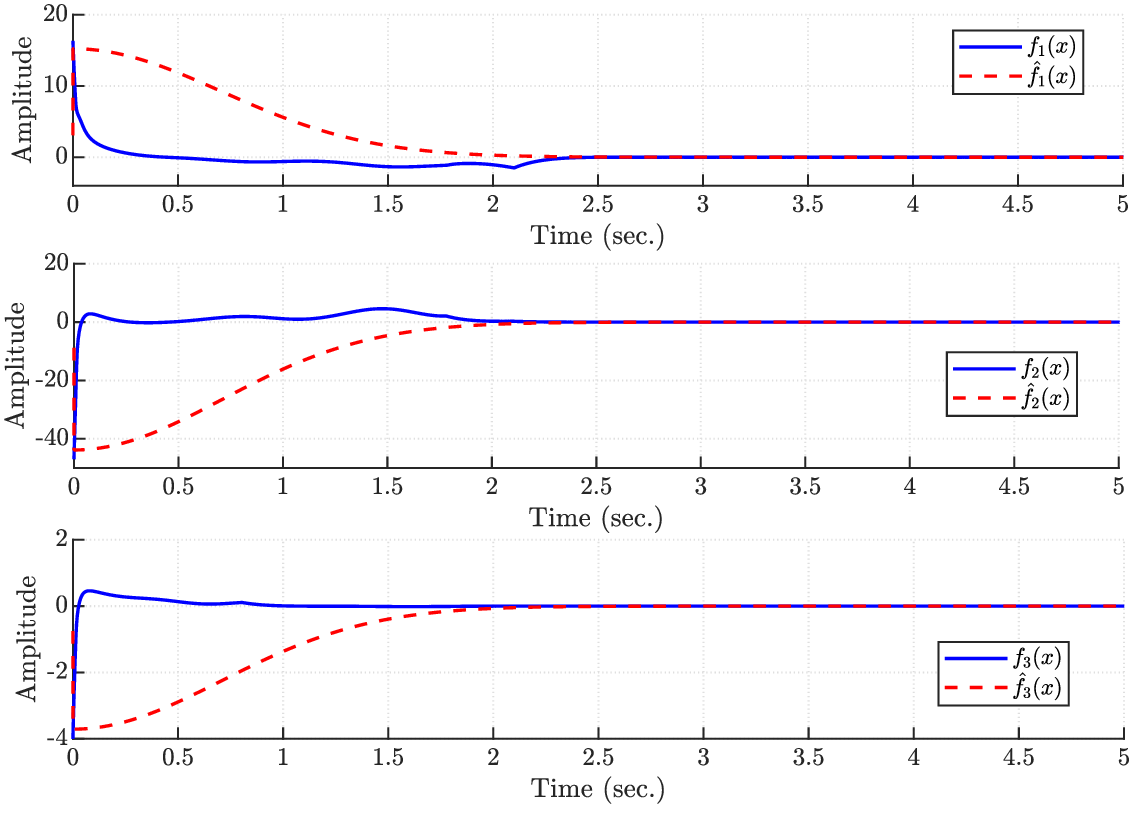}   
\caption{The unknown function  estimations  with  $N= 5$ data.} 
\label{fig:Estimated_GP_5}
\end{figure}
The plots depicting the  states, sliding surfaces, input controls, and function estimations using the online learning-based method from \cite{capone2019backstepping} and the FxT control learning are shown in Figs. \ref{fig:States_GP_5}--\ref{fig:Estimated_GP_5} for the first simulation (e.g., $N=5$) and in Figs. \ref{fig:States_GP_50}--\ref{fig:Estimated_GP_50} for the second simulation (e.g., $N=50$). The corresponding states in both cases are depicted in Figs. \ref{fig:States_GP_5} and \ref{fig:States_GP_50}. As can be seen, using more datasets leads to a drastic improvement in control performance compared to the smaller dataset while simultaneously requiring less data. This is because the data collected online using our approach always contributes to a significant reduction of the maximum model error, which in turn reduces the time derivative of the Lyapunov function and improves control performance. Figure \ref{fig:Sliding_GP_5} shows the sliding variable for the cases of $N=5$. The sliding surfaces illustrate how the control system is able to maintain the states within desired bounds more effectively as the dataset size increases. This improvement is consistent with the theoretical predictions, indicating that the online learning-based method enhances the robustness and stability of the control system. As indicated in Theorem \ref{theorem:2}, the settling time depends on the model error of the function $f_i$. When using less data, the settling time is greater than when using $N=50$. This is evident from the state trajectories, where the system with $N=50$ data points settles much faster to the desired states compared to the system with $N=5$ data points. The input control is depicted in Fig. \ref{fig:Inputs_GP_5}. These figures show the control efforts required to maintain system stability and achieve the desired performance. It can be observed that with a larger dataset ($N=50$), the control inputs are smoother and less aggressive compared to the smaller dataset ($N=5$), indicating a more efficient control strategy. The estimations for both scenarios are illustrated in Figs. \ref{fig:Estimated_GP_5} and \ref{fig:Estimated_GP_50}. These plots show the accuracy of the function estimations made by the online learning algorithm. With a larger dataset, the estimations are more accurate, which contributes to the overall improved performance of the control system.

\begin{figure}[!ht]
\centering
\includegraphics[scale=0.44]{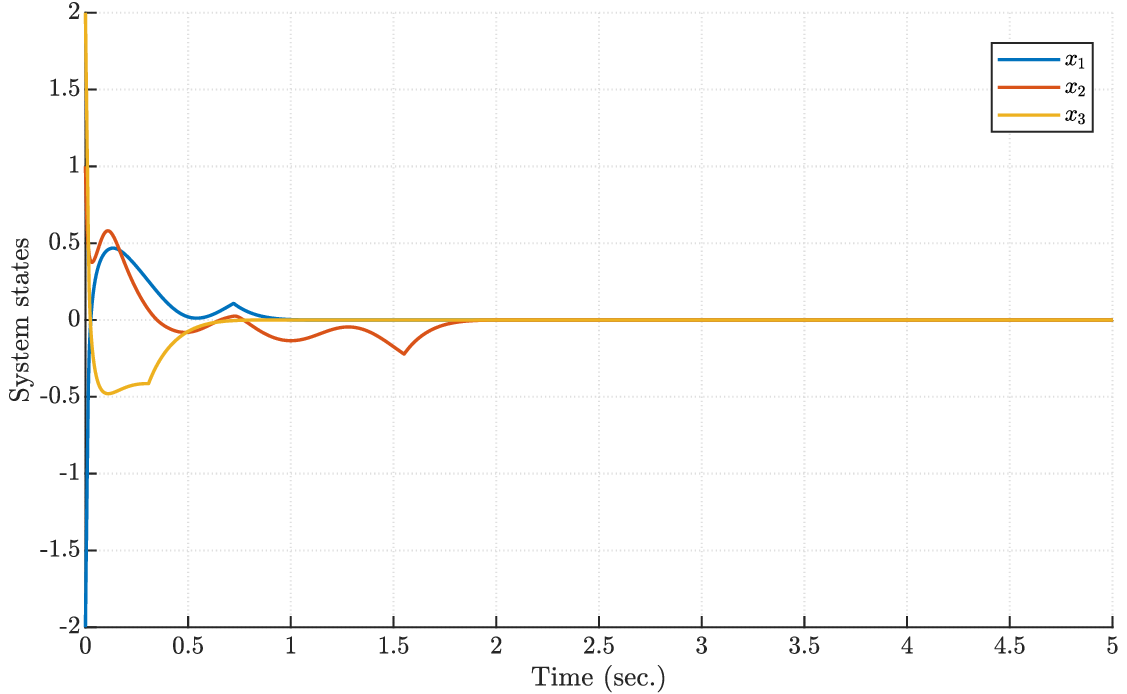}   
\caption{State variables  with $N= 50$ data.} 
\label{fig:States_GP_50}
\end{figure}
% \begin{figure}[!ht]
% \centering
% \includegraphics[scale=0.44]{Figs/Sliding_GP_50.eps}   
% \caption{Time response of the sliding variables  of controlled PMSM with $N= 50$ data.} 
% \label{fig:Sliding_GP_50}
% \end{figure}
% \begin{figure}[!ht]
% \centering
% \includegraphics[scale=0.44]{Figs/Inputs_GP_50.eps}   
% \caption{Time response of the input control  of controlled PMSM with $N= 50$ data.} 
% \label{fig:Inputs_GP_50}
% \end{figure}
\begin{figure}[!ht]
\centering
\includegraphics[scale=0.44]{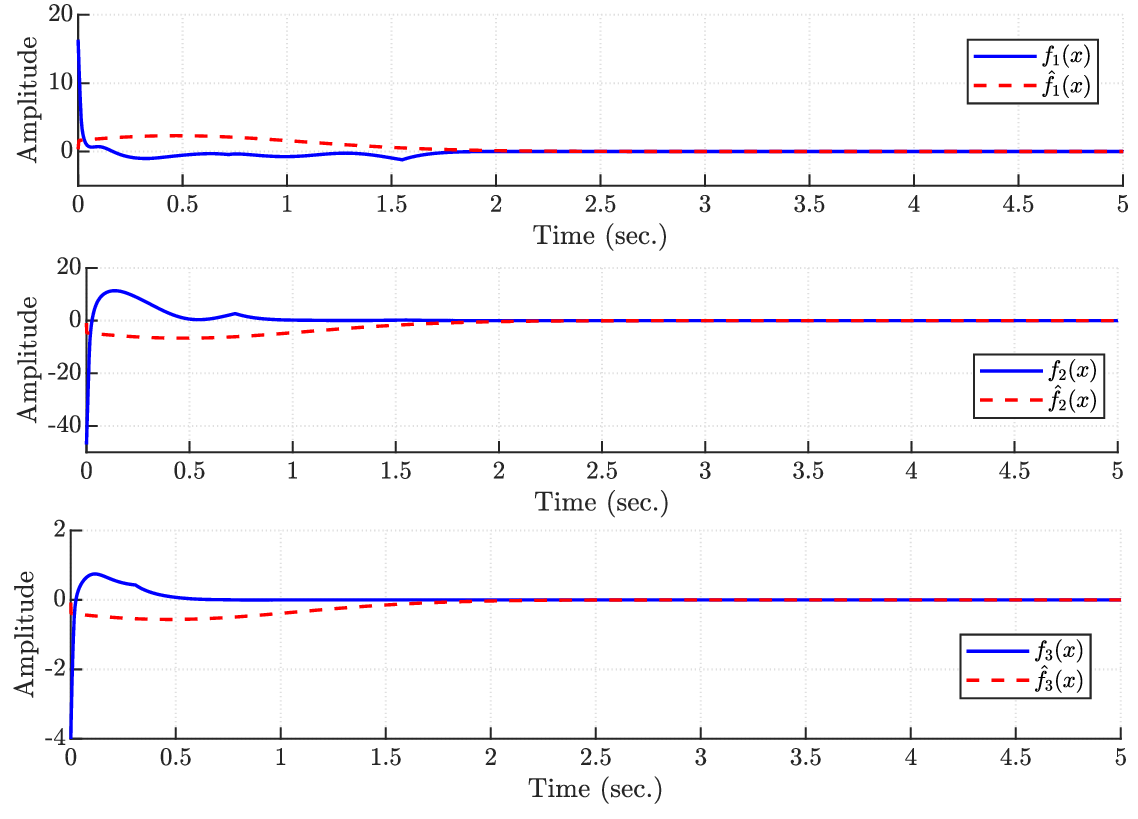}   
\caption{The unknown function estimation with  $N= 50$ data. } 
\label{fig:Estimated_GP_50}
\end{figure}

\section{Conclusions}
This paper addresses the problem of learning-based fixed-time control for a class of nonlinear systems subjected to matched perturbations. The first part focuses on designing an integral terminal sliding variable. Subsequently, a robust reaching phase is proposed, characterized by a single parameter necessary to achieve a Fixed-time FxT property. The estimation of the closed-loop system's settling time is provided. In the second part, GPs are employed to estimate certain parts of the system dynamics. 
 In the simulation section, we demonstrate that significant improvements in results are achievable with accurate data, comparable to those obtained when the model is known.
Future research directions could explore extending the applicability of this approach to other classes of continuous nonlinear systems.
% \section{Conclusions}\label{Sec:V}

 % \begin{thebibliography}{99}
 \balance
 \bibliographystyle{IEEEtran}
 \bibliography{bib}
 % \end{thebibliography}

\end{document}